\documentclass[12pt]{amsart}
\usepackage{amssymb,amsmath,mathrsfs,enumerate,mathtools,comment}
\usepackage{tikz-cd}
\usepackage[normalem]{ulem}
\usepackage[pdftex,bookmarks=true]{hyperref}
\usepackage{adjustbox}
\usepackage{hhline}
\usepackage{xcolor}
\usepackage{enumerate,graphicx}
\usepackage{filecontents}
\newtheorem*{theorem*}{Theorem}
\newtheorem*{lemma*}{Lemma}
\theoremstyle{plain}

\newtheorem{theorem}{Theorem}[section]

\newtheorem{lemma}[theorem]{Lemma}
\newtheorem{corollary}[theorem]{Corollary}
\theoremstyle{definition}

\newtheorem{example}[theorem]{Example}
\newtheorem{definition}[theorem]{Definition}
\newtheorem{remark}[theorem]{Remark}
\numberwithin{equation}{section}

\newcommand{\abs}[1]{\lvert#1\rvert}
\newcommand{\norm}[1]{\lVert#1\rVert}
\newcommand{\bigabs}[1]{\bigl\lvert#1\bigr\rvert}
\newcommand{\bignorm}[1]{\bigl\lVert#1\bigr\rVert}

\usepackage{array}

\usepackage{bbm}
\usepackage{float}

\newcommand{\N}{{\mathbb N}}
\newcommand{\E}{{\mathbb E}}
\newcommand{\bP}{{\mathbb P}}
\newcommand{\R}{{\mathbb R}}
\newcommand{\cR}{{\mathcal R}}
\newcommand{\cM}{{\mathcal M}}
\newcommand{\cX}{{\mathcal X}}

\title[Continuity and Consistency]{A note on continuity and asymptotic consistency  of  measures of risk and variability}

\author[N.~Gao]{Niushan Gao}
\address{Department of Mathematics, Toronto Metropolitan University, 350 Victoria Street, Toronto, Canada M5B 2K3}
\email{niushan@torontomu.ca}

\author[F.~Xanthos]{Foivos Xanthos}
\address{Department of Mathematics, Toronto Metropolitan University, 350 Victoria Street, Toronto, Canada M5B 2K3}
\email{foivos@torontomu.ca}

\thanks{The authors acknowledge support of NSERC Discovery Grants.}



\date{\today}

\begin{document}
\maketitle
\begin{abstract}
In this short note, we show that every   convex, order bounded above functional on a Fr\'echet lattice is automatically norm continuous. This improves a  results in \cite{RS06} and  applies to many deviation and variability measures. We also show that an order-continuous, law-invariant functional on an Orlicz space is strongly consistent everywhere, extending a result in \cite{KSZ14}. 
\end{abstract}

\section{Automatic Continuity}
Since its introduction  in the landmark paper Artzner et al.\ \cite{ArtznerDelbaenEberHeath1999}, the axiomatic theory of risk measures  has been a fruitful area of research.  Among many  topics, one particular direction is to investigate automatic continuity of risk measures.
In general, automatic continuity has   long  been an interesting research topic in mathematics and probably originates from the  fact that a real-valued convex function on an open interval is continuous. This well-known fact was later extended to the following  theorem for real-valued convex functionals on general Banach lattices. 

\begin{theorem*}[Ruszczy\'{n}ski and Shapiro \cite{RS06}] A real-valued, convex, monotone functional on a Banach lattice is norm continuous.
\end{theorem*}

Recall that a functional $\rho$ on a vector space $\cX$ is said to be \emph{convex} if $\rho(\lambda X+(1-\lambda )Y)\leq \lambda \rho(X)+(1-\lambda )\rho(Y) $ for  any $X,Y\in \cX$ and any $\lambda\in [0,1]$.   Recall also that a \emph{Banach lattice} $\cX$ is a vector lattice with a complete norm such that $\abs{X}\leq \abs{Y}$ in $\cX$ implies $\norm{X}\leq \norm{Y}$ (see  \cite{AB06} for standard terminology and facts regarding vector lattices).  A functional $\rho$ on a Banach lattice $\cX$ is said to be {\em increasing} if $\rho(X)\leq \rho(Y)$ whenever $X\leq Y$ in $\cX$. $\rho$ is said to be  \emph{decreasing} if $-\rho$ is increasing. $\rho$ is said to be \emph{monotone} if it is either increasing or decreasing.

The above celebrated result  of Ruszczy\'{n}ski and Shapiro has drawn extensive attention in optimization, operations research and risk management. We refer the reader to Biagini and Frittelli  \cite{BF}  for a version on Fr\'echet lattices and Farkas, Koch-Medina and Munari  \cite{FarkasKochMunari2014}  for further literature on automatic norm continuity properties.

With law invariance, other types of continuity properties beyond norm continuity can be established. The theorem below is striking. Recall first that a functional $\rho$ is said to be \emph{law invariant} if $\rho(X)=\rho(Y)$ whenever $X$ and $Y$ have the same distribution. Recall also that a functional $\rho$ on a set $\cX$ of random variables is said to have the \emph{Fatou property} if $\rho(X)\leq\liminf_n\rho(X_n)$ whenever $X_n\xrightarrow{o}X$ in $\cX$. Here $X_n\xrightarrow{o}X$ in $\cX$, termed as \emph{order convergence} in $\cX$\footnote{The term order convergence is originated from the theory of vector lattices (see e.g.\ \cite{AB06}). We note here that in our definition we do not require that $\cX$ is a vector lattice.}, is used in the literature to denote  dominated a.s.\ convergence in $\cX$, i.e.,    $X_n\stackrel{a.s.}{\longrightarrow}X$ and there exists $Y\in \cX$ such that $\abs{X_n}\leq Y$ a.s.\ for any $n\in\N$.    The Fatou property is therefore just order lower semicontinuity. 
\begin{theorem*}[Jouini et al.\ \cite{Joui06}]
A real-valued, convex, monotone, law-invariant functional on $L^\infty$ over a non-atomic probability space has the Fatou property. Consequently, it is $\sigma(L^\infty, L^1)$ lower semicontinuous and admits a dual representation via $L^1$.
\end{theorem*}

This theorem  was recently extended by Chen et al.\ \cite{CGLLa} to general rearrangement-invariant spaces. See \cite[Theorem 2.2., Theorem 4.3, Theorem 4.7]{CGLLa} for details. We also refer the reader to \cite{S13} for further interesting continuity properties  of law-invariant risk measures. 

\medskip
In this section, we aim at extending the above theorem of Ruszczy\'{n}ski and Shapiro on norm continuity of convex functionals.   
Specifically, we show that the monotonicity assumption can be significantly relaxed to the following notion on order boundedness. 
\begin{definition}
Let $\cX$ be a vector lattice. For $U,V\in \cX$ with $U\leq V$, the {\em order interval} $[U,V]$ is defined by $$[U,V]=\{X\in \cX:U\leq X\leq V\}.$$
A functional $\rho:\cX\rightarrow \R$ is said to be {\em order bounded above} if it is bounded above on all order intervals.
\end{definition}

Monotone functionals are easily seen to be order bounded above. While risk measures are usually assumed to be monotone, many important functionals used in  finance, insurance and other disciplines are not necessarily monotone. 

\begin{example}
General deviation measures were introduced in Rockafellar et al.\ \cite{RUZ06} as  functionals   $\rho:L^2 \rightarrow [0,+\infty]$ satisfying subadditivity, positive homogeneity, $\rho(X+c)=\rho(X)$ for every $X\in L^2$ and $c \in \mathbb{R}$, and $\rho(X)>0$ for nonconstant $X$. They are usually not monotone, but may be order bounded above. A specific example is standard deviation and semideviations. Recall that for a random variable $X\in L^2$, its standard deviation, upper and lower semideviations are given by $$\sigma(X)=\norm{X-\E[X]}_{L^2},\;\;\sigma_+(X)=\norm{(X-\E[X])^+}_{L^2},\;\;\sigma_-(X)=\norm{(X-\E[X])^-}_{L^2}, $$respectively. They are well known to be convex. It is also easy to see that they are neither increasing nor decreasing. We show that they are all order bounded above on $L^2$. Indeed, take any order interval $[U,V]\subset L^2$ and any $X\in[U,V]$. The desired order boundedness  property is immediate by the following inequalities.
\begin{align*}
U-\E[V]\leq &X-\E[X]\leq V-\E[U]\\
0\leq & (X-\E[X])^+\leq (V-\E[U])^+\\
 0\leq &(X-\E[X])^-\leq (U-\E[V])^-
\end{align*} 

\end{example}

\begin{example}
General variability measures were introduced in Bellini et al.\ \cite{BFWW22} as  law-invariant, positive-homogeneous functionals that vanish on constants. Many of them are also  order bounded above, although usually not monotone. In fact, all the three one-parameter families of variability measures in \cite[Section 2.3]{BFWW22} are easily seen to be order bounded above but not monotone. Let's discuss in details the  class of  inter-ES differences $\Delta_p^{\rm ES}$, $p\in(0,1)$, on $L^1$. Let $X\in L^1$ and $p\in (0,1)$. Recall the right and left expected shortfalls of $X$:
$$\mathrm{ES}_p(X)=\frac{1}{1-p}\int_p^1F_X^{-1}(t)\,\mathrm{d}t,\quad \mathrm{ES}_p^-(X)=\frac{1}{p}\int_0^pF_X^{-1}(t)\,\mathrm{d}t,$$
where $F_X^{-1}(t)=\inf\{x\in\R:\bP(X\leq x)\geq t\}$ is the left quantile function of $X$. 
The inter-ES difference $\Delta_p^{\rm ES}$ is defined by 
$$\Delta_p^{\rm ES}(X):=\mathrm{ES}_p(X)-\mathrm{ES}^-_{1-p}(X)=\mathrm{ES}_p(X)+\mathrm{ES}_p(-X).$$
$\Delta_p^{\rm ES}$ is clearly convex.
Take any order interval $[U,V]\subset L^1$ and any $X\in[U,V]$. By the monotonicity of expected shortfall, we get $$\Delta_p^{\rm ES}(X)\leq \mathrm{ES}_p(V)+\mathrm{ES}_p(-U)<\infty.$$ This proves that $\mathrm{ES}_p$ is order bounded above on $L^1$. Now suppose that $U\in L^1$ follows a uniform distribution on $[-1,1]$. Then $$U \leq 1,\quad \Delta_p^{\rm ES}(U)=2\mathrm{ES}_p(U)>0=\Delta_p^{\rm ES}(1),$$
$$-U \leq 1,\quad \Delta_p^{\rm ES}(-U)=2\mathrm{ES}_p(U)>0=\Delta_p^{\rm ES}(1).$$
Thus $\Delta_p^{\rm ES}$ is not monotone.
\end{example}

Our main result in this section is as follows. Recall  first  that a topological vector space $(\mathcal{X},\tau)$ is a  \emph{Fr\'echet lattice} if $\cX$ is a vector lattice and  $\tau$ is induced by a complete metric such that $0$ has a fundamental system of solid neighborhoods. A neighborhood $\mathcal{V}$ of $0 \in \cX$ is solid if $Y \in \mathcal{V}$ whenever  $X \in \mathcal{V}$ and $Y \in \cX$ satisfy that $|Y| \leq |X|$. All $L^p$ spaces and Orlicz spaces equipped with their natural norm are  Banach lattices and, in particular, are Fr\'echet lattices.

\begin{theorem}\label{theorem1}
Let $(\cX,\tau)$ be a  Fr\'echet lattice. 
Let $\rho:\mathcal{X} \rightarrow \mathbb{R}$ be a convex, order bounded above functional. Then $\rho $ is  $\tau$-continuous.
\end{theorem}

Recall that monotone functionals are order bounded above and Banach lattices are Fr\'echet. Thus Theorem \ref{theorem1} includes the preceding  theorem of Ruszczy\'{n}ski and Shapiro as a special case.

\begin{proof}[Proof of Theorem \ref{theorem1}]
Let $(X_n)$ and $X$ be such that $X_n\xrightarrow{{\tau}}X $ in $\cX$. We want to show that $\rho(X_n)\rightarrow \rho(X)$. Suppose otherwise that $\rho(X_n)\not\rightarrow \rho(X)$. By passing to a subsequence, we may assume that 
\begin{align}\label{eq1}
\abs{\rho(X_n)-\rho(X)}>\varepsilon_0\quad \text{for some }\varepsilon_0>0 \text{ and any }n\in\N.
\end{align}

Let $(\mathcal{V}_n)$ be a basis of $0$ for $\tau$ consisting of solid neighbourhoods such that $\mathcal{V}_{n+1}+\mathcal{V}_{n+1} \subseteq \mathcal{V}_n$ for each $n\in\N$. By passing to a further subsequence of $(X_n)$, we may assume  that $n\abs{X_{n}-X} \in  \mathcal{V}_n$ for any $n\geq 1$. Put $W_n=\sum_{i=1}^n i\abs{X_{i}-X} $. For any $n,m \in \mathbb{N}$, we have $$W_{n+m}-W_n= \sum_{i=n+1}^{n+m} i\abs{X_{i}-X} \in \mathcal{V}_{n+1}+\mathcal{V}_{n+2}+\cdots+\mathcal{V}_{n+m} \subseteq \mathcal{V}_{n}.$$
Thus $(W_n)$ is a $\tau$ Cauchy sequence. By the completeness of $\tau$,   there exists $Y \in \mathcal{X}$ such that $W_n \xrightarrow{{\tau}} Y$. Since $(W_n)$ is an increasing sequence, $W_n\uparrow Y$ (\cite[Theorem 3.46]{AB06}). In particular, it follows that $W_n\leq Y$ so that

\begin{align}\label{eq2}\abs{X_n-X}\leq \frac{1}{n}Y\quad\text{ for any }n\in\N.\end{align}
Moreover, since $\rho$ is order bounded above on $[X-Y,X+Y]=X+[-Y,Y]$, there exists a real number $M >0$  such that 
\begin{align}\label{eq3}
\rho(X+Z) \leq M \quad\text{ for any }Z \in [-Y,Y],\text{ i.e., whenever }\abs{Z}\leq Y.
\end{align}

Now fix any $\varepsilon>0$. Put  $N=\lfloor \frac{1}{\varepsilon}\rfloor+1$. By \eqref{eq2}, 
\begin{align}\label{eq4}
\frac{1}{\varepsilon}\abs{X_n-X}\leq Y\quad\text{ for any }n\geq N.
\end{align}
On one hand, by the convexity of $\rho$ and the following identity
\begin{align*}
X_n=&(1-\varepsilon)X+\varepsilon\Big(X+\frac{1}{\varepsilon}(X_n-X)\Big),
\end{align*}
we have
\begin{align*}  \rho(X_n) \leq& (1-\varepsilon) \rho(X)+\varepsilon \rho\Big(X+\frac{1}{\varepsilon}(X_n-X)\Big),
\end{align*}
implying that  
\begin{align*}
\rho(X_n)-\rho(X) \leq \varepsilon\Big( \rho\Big(X+\frac{1}{\varepsilon}(X_n-X)\Big)-\rho(X) \Big).
\end{align*}
This together with \eqref{eq4} and \eqref{eq3}  implies that
\begin{align}\label{eq5}
\rho(X_n)-\rho(X) \leq \varepsilon\big(M-\rho(X) \big)  \quad\text{ for any }n\geq N.
\end{align}
On the other hand,  by the convexity of $\rho$ and $$2X-X_n=(1-\varepsilon)X+\varepsilon\Big(X+\frac{1}{\varepsilon}(X-X_n)\Big), $$
we have as before that
\begin{align*}
\rho(2X-X_n)\leq &(1-\varepsilon) \rho(X)+\varepsilon \rho\Big(X+\frac{1}{\varepsilon}(X-X_n)\Big)\\
\leq & (1-\varepsilon) \rho(X)+\varepsilon M,
\end{align*}
for any $n\geq N$. In particular, $$\rho(2X-X_n)-\rho(X)\leq \varepsilon(M-\rho(X)) \quad\text{ for any }n\geq N.$$
By  $X=\frac{1}{2} X_n+\frac{1}{2}(2X-X_n)$ and the convexity of $\rho$,  we also get 
$$\rho(X) \leq \frac{1}{2} \rho(X_n)+\frac{1}{2} \rho(2X-X_n)$$
so that 
$$ \rho(X)-\rho(X_n) \leq \rho(2X-X_n)-\rho(X).$$
It follows that 
\begin{align}\label{eq6}
\rho(X)-\rho(X_n)\leq \varepsilon(M-\rho(X)) \quad\text{ for any }n\geq N.
\end{align}
Combining \eqref{eq5} and \eqref{eq6}, we have $$\abs{\rho(X)-\rho(X_n)}\leq \varepsilon(M-\rho(X)) \quad\text{ for any }n\geq N.$$
Hence, $\rho(X_n)\rightarrow \rho(X)$. This contradicts \eqref{eq1} and completes the proof.
\end{proof}

\section{Strong Consistency}
In this section, we discuss the strong consistency of estimating the risk $\rho(X)$ using estimates drawn from the empirical distributions. This problem has been studied for general convex risk measures on $L^p$ spaces and Orlicz hearts  in \cite{S13} and \cite{KSZ14}, respectively. We are motivated to study the case of general Orlicz spaces.


Throughout this section, fix a nonatomic probability space $(\Omega,\mathcal{F},\mathbb{P})$. 
Let $L^0$ be the space of all random variables on $\Omega$, with a.s.\ equal random variables identified as the same. 
Let $\cX$ be a subset of $L^0$. Denote  the set of distributions of all random variables in $\cX$ by $$\mathcal{M}(\mathcal{X})=\{ \mathbb{P} \circ X^{-1}  :  X \in \mathcal{X}\}.$$
Recall that a law-invariant functional $\rho$ on $\cX$  induces a natural mapping $\cR_\rho$ on $\cM(\cX)$ by 
$$\cR_\rho(\mathbb{P} \circ X^{-1})=\rho(X),\quad\text{ for any }X\in\cX.$$

Recall that a sequence $(X_n)$ of random variables is said to be stationary if for any $k, n \in \mathbb{N}$ and any $x_1,...,x_n \in \mathbb{R}$, it holds that $\mathbb{P}(X_1 \leq x_1,\cdots, X_n \leq x_n)=\mathbb{P}(X_{k+1}\leq x_1\cdots, X_{k+n}\leq x_n)$. 
Let $\mathcal{B}$ be  the Borel $\sigma$-algebra on $\mathbb{R}^{\mathbb{N}}$. A set $A \in \mathcal{F}$ is said to be invariant if there exists $B \in \mathcal{B}$ such that $A=\{(X_n)_{n \geq k} \in B\}$ for every $k\in \mathbb{N}$. 
A stationary sequence is said to be ergodic if every invariant set  has probability zero or one. Birkhoff's ergodic theorem states that the arithmetic averages of a stationary ergodic sequence $(X_n)$ converge a.s.\ to $\E[X_1]$ whenever $\E[\abs{X_1}]<\infty$.  See \cite[Section 6.7]{B91} for more facts regarding stationary and ergodic processes. 

Let $\cX$ be a subset of $L^0$ containing $L^\infty$. Take any $X\in \cX$. Let $(X_n)$ be a  stationary  ergodic sequence of random variables with the same distribution as $X$.   We denote the {\em empirical distribution} of $X$  arising from  $X_1,\dots,X_n$ by $$\widehat{m}_n=\frac{1}{n}  \sum_{i=1}^n \delta_{X_i};$$
here $\delta_x$ is the Dirac measure on $\R$ at $x$. Since $L^\infty\subset \cX$,  $\widehat{m}_n\in \cM(L^\infty)\subset \cM(\cX)$. This allows us to consider the corresponding {\em empirical estimate} for $\rho(X)$:
$$\widehat{\rho}_n:=\mathcal{R}_{\rho} (\widehat{m}_n);$$
We say that $\rho$ is {\em strongly consistent} at $X$  if for any   stationary  ergodic sequence of random variables with the same distribution as $X$, $$\widehat{\rho}_n=\mathcal{R}_{\rho} (\widehat{m}_n) \xrightarrow{{a.s.}} \rho(X).$$
We refer to  \cite{S13, KSZ14,KSZ17} and the references therein for  literature on strong (and weak) consistency of risk measures.
In particular, we are interested in the following result.

\begin{theorem*}[Kr\"{a}tschmer et al.\ {\cite[Theorem 2.6]{KSZ14}}]A  norm-continuous, law-invariant functional on an Orlicz heart   is strongly consistent everywhere.
\end{theorem*}

We remark that this theorem was originally stated for real-valued, law-invariant, convex risk measures. In their definition  in \cite{KSZ14}, convex risk measures   are assumed to be monotone, and thus are norm continuous by  the aforementioned theorem of Ruszczy\'{n}ski and Shapiro. A quick examination of the proof of \cite[Theorem 2.6]{KSZ14} shows that norm continuity and law invariance are the only ingredients of the functional used.

Let's  recall the definitions of Orlicz spaces and hearts. A function $\Phi:[0,\infty)\to[0,\infty)$ is called an {\em Orlicz function} if it is non-constant, convex, increasing,  and $\Phi(0)=0$.   The {\em Orlicz space} $L^\Phi$ is the space of all $X\in L^0$ such that the {\em Luxemburg norm} is finite:
\[
\norm{X}_\Phi:=\inf\Big\{\frac{1}{\lambda}:\lambda>0 \text{ and } \E\left[\Phi\big(\lambda\abs{X}\big)\right]\leq 1 \Big\}<\infty.
\]
The {\em Orlicz heart} $H^\Phi$ is a subspace of $L^\Phi$ defined by \[
H^\Phi := \left\{X\in L^0 : \E\left[\Phi\big(\lambda|X|\big)\right]<\infty \text{ for any } \lambda >0\right\}.
\]
We refer to \cite{ES92} for standard terminology and facts on Orlicz spaces. Risk measures on Orlicz spaces have been studied extensively; see, e.g., \cite{BCc, BR12,BFG11,  GLMX18, GLX19, GMX20,GX18} and the references therein.

%

The above theorem in conjunction with Theorem \ref{theorem1} immediately yields the following result, which improves \cite[Theorem 2.6]{KSZ14}.

\begin{corollary}
A  convex, law-invariant, order bounded above functional on an Orlicz heart is strongly consistent everywhere.
\end{corollary}

The above theorem of Kr\"{a}tschmer et al is essentially due to the fact that for any $X\in H^\Phi$, and for a.e.\ $\omega\in \Omega$, there exist a random variable $X^\omega$ on $\Omega$ with same distribution as $X$ and  a sequence of  random variables $(X_n^\omega)$ on $\Omega$ with distributions $\widehat{m}_n(\omega)$'s  such that 
\begin{align}\label{eq7}
\norm{X_n^\omega-X^\omega}_\Phi\rightarrow 0.
\end{align}
This, however, does not hold for arbitrary random variables in a general Orlicz space $L^\Phi$. Specifically, when $\Phi$ fails the $\Delta_2$-condition, there exists $X\in L^\Phi\backslash H^\Phi$. For this $X$, \eqref{eq7} must fail: $X_n^\omega$ takes only at most $n$ values and thus is  a simple random variable lying in $H^\Phi$; therefore, \eqref{eq7} would imply $X\in H^\Phi$ as well.

%

We extend the theorem of  Kr\"{a}tschmer et al.\ as follows.  Recall first that on a set $\cX\subset L^0$, a functional $\rho:\cX\rightarrow  \R$ is said to be {\em order continuous} at $X\in\cX$ if $\rho(X_n)\rightarrow \rho(X)$ whenever $X_n\xrightarrow{o}X $  in $\cX$.  In the literature, order continuity is also termed as the {\em Lebesgue property}.

\begin{theorem}\label{theorem2}
An order-continuous, law-invariant  functional on an Orlicz space  is strongly consistent  everywhere.
\end{theorem}

For the proof of Theorem \ref{theorem2}, we need to establish a few technical lemmas, which along the way also reveal why order continuity is the most natural condition for general Orlicz spaces. For an Orlicz function $\Phi$, the Young class is defined by
\[
Y^\Phi := \left\{X\in L^0 : \E\left[\Phi\left(|X|\right)\right]<\infty\right\}.
\]
It is easy to see that $H^\Phi\subset Y^\Phi\subset L^\Phi$. As in \cite{KSZ14}, we use the term $\Phi$-weak topology in place of the $\Phi(\abs{\cdot})$-weak topology on $\mathcal{M}(Y^{\Phi})$ for brevity. This topology is metrizable (see e.g.\  \cite[Corollary A.45]{FolmerSchied2011}). For the special case where $\Phi(x)=\frac{x^p}{p}$ for some $1 \leq p<\infty$, the $\Phi$-weak topology  generates the Wasserstein metric of order $p$ (see e.g.\ \cite[Theorem 7.12]{VIL21}). Moreover, for a sequence $(\mu_n)\subset \mathcal{M}(Y^{\Phi})$  and $ \mu_0 \in \mathcal{M}(Y^{\Phi})$,   $(\mu_n)$ {\em converges  $\Phi$-weakly } to $\mu_0$, written as $\mu_n\xrightarrow{\Phi\text{-weakly}}\mu_0$, iff 
\[
\ \mu_n\xrightarrow{\text{weakly}}\mu \ \ \mbox{and} \ \  \int \Phi(|x|) \mu_{n}(dx) \rightarrow \int \Phi(|x|) \mu_{0}(dx).\]

The following Skorohod representation for $\Phi$-weak convergence is a general order version of  \cite[Theorem 3.5]{KSZ14} and \cite[Theorem 6.1]{KSZ17} beyond the Orlicz heart and without any restrictions on $\Phi$.

\begin{lemma}\label{lem1}
\begin{enumerate}
\item[(i)] Let $(\mu_n)$ be a sequence in $\mathcal{M}(Y^{\Phi})$ that converges $\Phi$-weakly to some $\mu_0 \in \mathcal{M}(Y^{\Phi})$. Then there exist  a subsequence $(\mu_{n_k})$ of $(\mu_n)$, a sequence $(X_k)$ in $Y^{\Phi}$  and $X\in Y^{\Phi}$ such that  $X_k$ has distribution $\mu_{n_k}$ for each $k \in \mathbb{N}$,  $X$ has distribution  $\mu_0$, and $X_k \xrightarrow{o} X$ in $Y^{\Phi}$.
\item[(ii)] Let $(X_n)$  in $Y^{\Phi}$  and $X \in Y^{\Phi}$ be such that $X_n\xrightarrow{o}X \text{ in } Y^{\Phi}$. Then $\mu_n\xrightarrow{\Phi\text{-weakly}}\mu_0$, where $\mu_n$'s are the distributions of $X_n$'s and $\mu_0$ is the distribution of $X$, respectively. 
\end{enumerate}
\end{lemma}

\begin{proof}
We start with the following observation. Since $\Phi$ is continuous and increasing, for any sequence $(X_n)$ in $Y^{\Phi}$ we have 

\begin{align} \label{eq8}
\mathbb{E}\left[\Phi \Big(\sup_{n \in \mathbb{N}} |X_n|\Big)\right]=\mathbb{E}\left[\sup_{n \in \mathbb{N}} \Phi(|X_n|)\right]
\end{align}

(i). Take $(\mu_n)$ in $\mathcal{M}(Y^{\Phi})$ that converges $\Phi$-weakly to $\mu_0 \in \mathcal{M}(Y^{\Phi})$. Since the probability space is nonatomic, the classical Skorohod representation yields $(X_n)\subset Y^\Phi$ and $X\in Y^\Phi$ such that $X_n \sim \mu_n$ for every $n\in\N$, $X \sim \mu_0$ , and  $X_n\xrightarrow{a.s.}X$. Clearly,
\begin{equation}\label{eq9}
\mathbb{E}[\Phi(|X|)]= \int \Phi(|x|) \mu_{0}(dx)= \lim_n\int \Phi(|x|) \mu_{n}(dx)=\lim_n\mathbb{E}[\Phi(|X_n|)]<\infty.
\end{equation}
Since $\Phi$ is continuous, we also have that $\Phi(\abs{X_n}) \xrightarrow{{a.s.}} \Phi(\abs{X}) $. This combined with \eqref{eq9} yields (see \cite[Theorem 31.7]{AB98}) that $$\bignorm{\Phi(\abs{X_n}) -\Phi(\abs{X})}_{L^1}\rightarrow 0.$$
Passing to a subsequence we may assume that $$\sum_{n=1}^\infty\bignorm{ \Phi(|X_n|)-\Phi(\abs{X})}_{L^1}<\infty$$ so that $$\sum_{n=1}^\infty \bigabs{\Phi(|X_n|)-\Phi(\abs{X})}\in L^1.$$ 
In particular, $$\sup_{n\in\N}\bigabs{\Phi(|X_n|)-\Phi(\abs{X})}\in L^1.$$
It follows from $\Phi(\abs{X_n})\leq \bigabs{\Phi(|X_n|)-\Phi(\abs{X})}+\Phi(\abs{X})$ that $\sup_{n\in\N }\Phi(\abs{X_n}) \in L^1 $.
Hence, by \eqref{eq8}, $\mathbb{E}[\Phi (\sup_{n \in \mathbb{N}} |X_n|)]=\mathbb{E}[\sup_{n \in \mathbb{N}} \Phi(|X_n|)]<\infty$. That is, $\sup_{n \in \mathbb{N}} \abs{X_n}\in Y^\Phi$; equivalently, $(X_n)$ is dominated in $Y^\Phi$.
In particular, we have $X_n \xrightarrow{o}X$ in $Y^{\Phi}$

(ii). Let $(X_n)$ be such that $X_n\xrightarrow{o}X$ in $Y^{\Phi}$ and $\mu_n$ be  the distribution of $X_n$ for each $n$, $\mu_0$ be the distribution of $X$. We clearly have $ \mu_n\xrightarrow{\text{weakly}}\mu $ and by the continuity of $\Phi$, we get $\Phi(|X_n|) \xrightarrow{{a.s.}} \Phi(|X|)$. Since $(X_n)$ is dominated in $Y^\Phi$, $\sup_{n\in\N}\abs{X_n}\in Y^\Phi$. Thus in view of \eqref{eq8}, we get $$\mathbb{E}\left[\sup_{n \in \mathbb{N}} \Phi(|X_n|)\right]=\mathbb{E}\left[\Phi \Big(\sup_{n \in \mathbb{N}} |X_n|\Big)\right]<\infty,$$
i.e.,  $\sup_{n \in \mathbb{N}} \Phi(|X_n|)\in L^1$. By the dominated convergence theorem, we get 

$$ \int \Phi(|x|) \mu_{0}(dx)=\mathbb{E}[\Phi(|X|)]= \lim_n\mathbb{E}[\Phi(|X_n|)]=\lim_n\int \Phi(|x|) \mu_{n}(dx). $$
This proves $\mu_n\xrightarrow{\Phi\text{-weakly}}\mu_0$.
\end{proof}

The lemma below reveals the essential and natural role of order continuity.

\begin{lemma}\label{lem2}
Let $\rho:Y^\Phi\rightarrow \R$ be law invariant. The following are equivalent.
\begin{enumerate}
\item[(i)] $\mathcal{R}_\rho$ is continuous on $\mathcal{M}(Y^{\Phi})$ with the $\Phi$-weak topology. 
\item[(ii)] $\rho$ is order continuous on  $Y^\Phi$.
\end{enumerate}
\end{lemma}

\begin{proof}
(ii)$\implies$(i). Suppose that (ii) holds but  (i)  fails. Recall that the $\Phi$-weak topology is metrizable. Thus  we can find a sequence $(\mu_n)$ and $\mu_0$ in $\mathcal{M}(Y^{\Phi})$ such that $\mu_n\xrightarrow{\Phi\text{-weakly}}\mu_0$ but $\mathcal{R}_\rho (\mu_n) \not\rightarrow \mathcal{R}_{\rho}(\mu_0)$. Passing to a subsequence, we may assume that 
\begin{equation}\label{eq10} |\mathcal{R}_{\rho} (\mu_n)-\mathcal{R}_{\rho}(\mu_0)| \geq \varepsilon_0,
\end{equation}
for some $\varepsilon_0>0$ and all $n 
 \in \mathbb{N}$. By Lemma \ref{lem1}(i), there exist a subsequence $(\mu_{n_k})$ of $(\mu_n)$, a sequence $(X_k)$ in $Y^{\Phi}$  and $X\in Y^{\Phi}$ such that  $X_k$ has distribution $\mu_{n_k}$ for each $k \in \mathbb{N}$,  $X$ has distribution  $\mu_0$, and $X_k \xrightarrow{o} X$ in $Y^{\Phi}$. (ii) implies that $$\mathcal{R}_{\rho}(\mu_{n_k}) =\rho(X_k)\rightarrow \rho(X)=\mathcal{R}_{\rho}(\mu_0).$$
This contradicts \eqref{eq10} and proves (ii)$\implies$(i). The reverse implication (i)$\implies$(ii)  is immediate by Lemma \ref{lem1}(ii).
\end{proof}

We now present the proof of Theorem \ref{theorem2}.
\begin{proof}[Proof of Theorem \ref{theorem2}]
Suppose that  $\rho:L^\Phi\rightarrow \R$ is law invariant and order continuous. Take any $X\in L^\Phi$ and any  stationary   ergodic sequence of random variables with the same distribution as $X$. Denote by $\mu_0$ their distribution.
Let  $\lambda>0$ be such that $\mathbb{E}[\Phi(\lambda\abs{X})]<\infty$. Put $\Phi_{\lambda}(\cdot):=\Phi(\lambda {\cdot})$.
Arguing similarly as in the proof of \cite[Theorem 2.6]{KSZ14}, by applying Birkhoff's ergodic theorem, one obtains a measurable subset $\Omega_0$ of $\Omega $ such that $\bP(\Omega_0)=1$ and for every  $\omega\in \Omega_0$, 
$$\widehat{m}_n(\omega)\xrightarrow{\Phi_\lambda\text{-weakly}}\mu_0.$$
Since $\rho$ is order continuous on  $L^{\Phi}$ and $Y^{\Phi_\lambda}\subset L^\Phi$, $\rho$ is also order continuous on $Y^{\Phi_\lambda}$. By Lemma \ref{lem2}, $\mathcal{R}_\rho$ is continuous on $\mathcal{M}(Y^{\Phi_\lambda})$ with the $\Phi_\lambda$-weak topology. 
Thus $\widehat{\rho}_n(\omega)=\cR_\rho(\widehat{m}_n(\omega)) \rightarrow\cR_\rho(\mu_0)= \rho(X)$ for every  $\omega\in \Omega_0$. This proves that $\rho$ is strongly consistent at $X$. 
\end{proof}

\begin{remark}
In our definition of Orlicz spaces, we do not allow the Orlicz function to take the  $\infty$ value, which excludes $L^\infty$ from the above considerations.  However, Theorem \ref{theorem2} remains true for $L^\infty$. Let $\rho:L^\infty \rightarrow \mathbb{R}$ be an order continuous, law invariant functional. Take any $X\in L^\infty$ and any  stationary   ergodic sequence of random variables with the same distribution as $X$. Denote by $\mu_0$ their distribution. By Birkhoff's ergodic theorem and an application of Theorem 6.6 in \cite[Chapter 1]{P67},  there exists a measurable subset $\Omega_0$ of $\Omega $ such that $\bP(\Omega_0)=1$  and   $\widehat{m}_n(\omega)\xrightarrow{\text{weakly}}\mu_0$ for any $\omega \in \Omega_0$.   One may assume further that $\abs{X(\omega)} \leq \norm{X}_{\infty}$ for any $\omega \in \Omega_0$. Fix any $\omega \in \Omega_0$. The classical Skorohod representation yields $(X_n^{\omega})\subset L^{\infty}$ and $X\in L^{\infty}$ such that $X_n^{\omega} \sim \widehat{m}_n(\omega)$ for every $n\in\N$, $X \sim \mu_0$ , and  $X_n^{\omega}\xrightarrow{a.s.}X$. We may assume that $\abs{X_n^{\omega}} \leq \norm{X}_{\infty}$ on $\Omega$ for every $n\geq 1$.  It follows that $X_n^{\omega} \xrightarrow{{o}} X$ in $L^\infty$. Hence,  by order continuity of $\rho$, we get  $\widehat{\rho}_n(\omega)=\cR_\rho(\widehat{m}_n(\omega))=\rho(X_n^{\omega}) \rightarrow \rho(X)$. This proves  that $\rho$  is strongly consistent on $L^\infty$.
\end{remark}

Order continuity of law-invariant functionals on Orlicz spaces is generally stronger than norm continuity. In the following, we show that it is satisfied by a large class of risk measures, namely, spectral risk measures. Spectral risk measures were introduced in  Acerbi \cite{A02} and includes many important risk measures such as the expected shortfall.   Let  $\phi$ be a nonnegative and nondecreasing function   such that $\int_0^1 \phi(t)dt=1$ ($\phi$ is called a spectrum).   The associated spectral risk measure   is defined by 
$$\rho_\phi(X)=\int_{0}^1 \phi(t)F^{-1}_X(t)\,\mathrm{d}t,\quad X\in L^1$$
where $F_X^{-1}(t)=\inf\{x \in \mathbb{R}: F(x) \geq t\}$ is the left quantile function of $X$. It is known that $\rho_\phi$ takes values in $(-\infty,\infty]$ and is convex, monotone and lower semicontinuous with respect to the $L^1$ norm (see e.g.\ \cite[Lemma C.1]{AL23}). For spectral risk measures, the empirical estimator $\widehat{\rho}_n$ has the form of an $L$-statistic and the strong consistency can be studied using tools from the theory of $L$-statistics  (see e.g.\ \cite{T14}).  Below we give a simple proof of the strongly consistency of $\rho_{\phi}$ in the Orlicz space framework based on  Theorem \ref{theorem2}.

\begin{corollary}
Let $\phi$ be a spectrum function such that $\rho_{\phi}$ is real-valued on $L^\Phi$. Then $\rho_{\phi}$ is order continuous on $L^{\Phi}$ and is thus strongly consistent everywhere on $L^\Phi$.
\end{corollary}

\begin{proof}
Suppose that  $X_n \xrightarrow{{a.s}} X$ and there exists $Y \in L^\Phi$ such that $\abs{X_n}\leq Y$ for every $n\in\N$. It is well-known that  $F_{X_n}^{-1} \xrightarrow{{a.s.}} F_X^{-1} $ on $(0,1)$ (with the Lebesgue measure). Hence, $$F_{X_n}^{-1}\phi \xrightarrow{{a.s.}} F_X^{-1}\phi.$$
Next,   note that since $-Y\leq X_n\leq Y$, $F_{-Y}^{-1}\leq F_{X_n}^{-1}\leq F_{Y}^{-1}$ on $(0,1)$. Thus 
$$F_{-Y}^{-1}\phi\leq F_{X_n}^{-1}\leq F_{Y}^{-1}\phi\quad\text{for every }n\in\N.$$
Since $\rho$ is real-valued on $L^\Phi$, $F_{Y}^{-1}\phi\in L^1$ and $F_{-Y}^{-1}\phi\in L^1$. Thus, by the dominated convergence theorem, $$\rho_{\phi}(X_n)=\int_{0}^1 \phi(t)F^{-1}_{X_n}(t)\,\mathrm{d}t \rightarrow\int_{0}^1 \phi(t)F^{-1}_X(t)\,\mathrm{d}t= \rho_{\phi}(X).$$
This proves that $\rho_\phi$ is order continuous on $L^\phi$. 
The strongly consistency   follows from   Theorem \ref{theorem2}.
\end{proof}

We end this note with the following remark that improves the implication (b)$\implies $(a) in \cite[Theorem 2.8]{KSZ14} due to our Theorem \ref{theorem1}.

\begin{corollary}
Suppose that $\Phi$ satisfies the $\Delta_2$-condition. Let $\rho$ be any   convex, law-invariant, order bounded above functional on $L^\Phi$. Then $\cR_\rho$ is continuous on $\cM(L^\Phi)$ for the $\Phi$-weak topology.
\end{corollary}
 
\begin{proof}
By Theorem \ref{theorem1}, $\rho$ is norm continuous. When $\Phi$ satisfies the $\Delta_2$-condition, order convergence in $L^\Phi$ implies norm convergence. Thus $\rho$ is also order continuous. Under the $\Delta_2$-condition, we also have $H^\Phi=Y^\Phi=L^\Phi$. Now apply Lemma \ref{lem2}.
\end{proof}

\end{document}